\documentclass[prl,aps,twocolumn,twoside,superscriptaddress,longbibliography,nofootinbib]{revtex4-1}

\usepackage{graphicx}
\usepackage[pdftex,colorlinks=true,linkcolor=blue,citecolor=blue,urlcolor=black]{hyperref}
\usepackage{amsmath, amsthm, amssymb}
\usepackage{subfigure}
\usepackage{comment}
\usepackage{url}
\usepackage{xcolor}

\newcommand{\C}{\mathbb{C}}

\newcommand{\F}{\mathbb{F}}

\newcommand{\ip}[2]{\langle #1|#2 \rangle}

\newcommand{\bracket}[3]{\langle #1|#2|#3 \rangle}
\newcommand{\sm}[1]{\left( \begin{smallmatrix} #1 \end{smallmatrix} \right)}

\DeclareMathOperator{\tr}{tr}

\DeclareMathOperator{\vecc}{vec}

\newcommand{\be}{\begin{equation}}
\newcommand{\ee}{\end{equation}}
\newcommand{\bea}{\begin{eqnarray}}
\newcommand{\eea}{\end{eqnarray}}
\newcommand{\bes}{\begin{equation*}}
\newcommand{\ees}{\end{equation*}}
\newcommand{\beas}{\begin{eqnarray*}}
\newcommand{\eeas}{\end{eqnarray*}}

\newtheorem{thm}{Theorem}
\newtheorem*{thm*}{Theorem}

\newtheorem{lem}[thm]{Lemma}
\newtheorem*{lem*}{Lemma}

\newtheoremstyle{definition}
   {}{}{}{0pt}{\bfseries}{.}{.5cm}
   {{\thmname{#1 }}{\thmnumber{#2}}{\thmnote{ (#3)}}}

\theoremstyle{definition}

\newtheorem{algm}{Algorithm}

\newcommand{\boxdfn}[2]{
\begin{figure}[h]
\begin{center}
\noindent \framebox{
\begin{minipage}{8cm}
\begin{algm}[#1]
\ \\[10pt]
#2
\end{algm}
\end{minipage}
}
\end{center}
\end{figure}
}

\input{Qcircuit.tex}

\begin{document}

\title{Learning stabilizer states by Bell sampling}

\author{Ashley Montanaro}
\affiliation{School of Mathematics, University of Bristol, UK}
\email{ashley.montanaro@bristol.ac.uk}

\begin{abstract}
We show that measuring pairs of qubits in the Bell basis can be used to obtain a simple quantum algorithm for efficiently identifying an unknown stabilizer state of $n$ qubits. The algorithm uses $O(n)$ copies of the input state and fails with exponentially small probability.
\end{abstract}

\maketitle


It is well-known and follows from Holevo's theorem~\cite{holevo73} that approximately determining an arbitrary quantum state $\ket{\psi}$ of $n$ qubits requires exponentially many (in $n$) copies of $\ket{\psi}$. One way of circumventing this problem is to relax the notion of what it means to determine $\ket{\psi}$ (e.g.\ by requiring only that we are able to predict the result of ``most'' measurements on $\ket{\psi}$, according to some probability distribution~\cite{aaronson07}); another way is to restrict the class of states to be determined to some class which can be described efficiently. In this setting, we are given a quantum system that is promised to be in a state picked from some family of quantum states, and are asked to determine its state, exactly or approximately.

One example where efficient identification can be achieved is the class of states well approximated by a matrix product state~\cite{cramer10}. Another example, on which we will focus here, is the class of stabilizer states. Aaronson and Gottesman described an efficient procedure for identifying an unknown stabilizer state $\ket{\psi}$ of $n$ qubits~\cite{aaronson08}. One variant of their algorithm uses $O(n^2)$ copies of $\ket{\psi}$. In this variant, all measurements are performed on single copies of $\ket{\psi}$. Another variant uses only $O(n)$ copies of $\ket{\psi}$, but is based on collective measurements across all these copies. This second algorithm is information-theoretically optimal: as there are $2^{\Theta(n^2)}$ stabilizer states on $n$ qubits~\cite{aaronson04a}, identifying $\ket{\psi}$ requires $\Omega(n)$ copies of $\ket{\psi}$ by Holevo's theorem~\cite{holevo73}.

In related work, Low has shown that an unknown element $U$ of the Clifford group on $n$ qubits can be identified with $O(n^2)$ uses of $U$, or even only $O(n)$ if $U^\dag$ is also available~\cite{low09}. Rocchetto has shown that an unknown stabilizer state can be learned efficiently in the PAC model~\cite{rocchetto17}.

Here we will prove the following result:

\begin{thm}
\label{thm:main}
There is a quantum algorithm which identifies an unknown stabilizer state $\ket{\psi}$ of $n$ qubits given access to $O(n)$ copies of $\ket{\psi}$. The algorithm makes collective measurements across at most two copies of $\ket{\psi}$ at a time, runs in time $O(n^3)$ and fails with probability exponentially small in $n$.
\end{thm}

The number of copies of $\ket{\psi}$ used by this algorithm thus matches that of Aaronson and Gottesman's collective-measurement algorithm~\cite{aaronson08}, but the algorithm acts on a smaller number of copies at a time. In addition, the measurements made by the algorithm across pairs of copies of $\ket{\psi}$ are simple to implement: they are based on measuring pairs of corresponding qubits of $\ket{\psi}^{\otimes 2}$ in the Bell basis. This is reminiscent of the algorithm of~\cite{harrow13} for testing product states, where the measurement performed across pairs of qubits was the swap test.

An alternative algorithm for identifying an unknown graph state (a subclass of stabilizer states) on $n$ qubits using $O(n)$ copies has been presented in independent work of Zhao, P\'erez-Delgado and Fitzsimons~\cite{zhao16}. Their algorithm has some structural similarities to the algorithm of the present paper.


\subsection{Preliminaries}

We will use the matrices
\begin{align*}
\sigma_{00}& := \begin{pmatrix}1 & 0\\ 0 & 1\end{pmatrix},\, \sigma_{01} := \begin{pmatrix}0 & 1\\ 1 & 0\end{pmatrix},\,  \sigma_{10} := \begin{pmatrix}1 & 0\\ 0 & -1\end{pmatrix},\\
\sigma_{11}& := \sigma_{10} \sigma_{01} = \begin{pmatrix}0 & 1\\ -1 & 0\end{pmatrix},
\end{align*}
which are the Pauli matrices up to applying $-i$ to $\sigma_{11}$, and the Bell basis, i.e.\ the ordered basis of $\C^4$ which we define by
\begin{align*} \ket{\sigma_{00}} &:= \frac{1}{\sqrt{2}}(\ket{00} + \ket{11}),\, \ket{\sigma_{01}} := \frac{1}{\sqrt{2}}(\ket{01} + \ket{10}),\\
\ket{\sigma_{10}} &:= \frac{1}{\sqrt{2}}(\ket{00} - \ket{11}),\, \ket{\sigma_{11}} := \frac{1}{\sqrt{2}}(\ket{01} - \ket{10}).
\end{align*}
The notation is supposed to highlight the fact that $\ket{\sigma_i} = \vecc(\sigma_i)/\sqrt{2}$, where $\vecc$ is the linear map defined by $\vecc(\ket{x}\bra{y}) = \ket{x}\ket{y}$ for computational basis states $x$, $y$. The $\vecc$ operator preserves inner products: $\ip{\vecc(A)}{\vecc(B)} = \tr A^\dag B$. For $s \in \{0,1\}^{2n}$, we write $\sigma_s := \sigma_{s_1 s_2} \otimes \dots \otimes \sigma_{s_{2n-1} s_{2n}}$, $\ket{\sigma_s} := \ket{\sigma_{s_1 s_2}} \dots \ket{\sigma_{s_{2n-1} s_{2n}}}$. Up to multiplying by $-1$, $\sigma_s \sigma_t = \sigma_{s \oplus t}$.

Measurement in the Bell basis can be implemented by applying the circuit
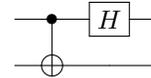
\[
\Qcircuit @C=1em @R=.7em {
 & \ctrl{1} & \gate{H} & \qw \\
 & \targ & \qw & \qw
} 
\]
and measuring in the computational basis. Given a pure state of $2n$ qubits divided into systems $A_1,\dots,A_n$, $B_1,\dots,B_n$, we call the operation of measuring each pair $A_i B_i$ of qubits in the Bell basis {\em Bell sampling}. Each such measurement returns a $2n$-bit string.

For any state $\ket{\psi}$, let $\ket{\psi^*}$ denote the complex conjugate (taken in the computational basis).

\begin{lem}
\label{lem:bsamp}
Let $\ket{\psi}$ be a state of $n$ qubits. Bell sampling on $\ket{\psi}^{\otimes 2}$ returns outcome $r$ with probability
\[ \frac{|\bracket{\psi}{\sigma_r}{\psi^*}|^2}{2^n}. \]
\end{lem}

\begin{proof}
We have $\ket{\psi}\ket{\psi} = \vecc(\ket{\psi}\bra{\psi^*})$, so $|\ip{\sigma_r}{\psi}\ket{\psi}|^2 =$ $2^{-n} |\tr \sigma_r^\dag \ket{\psi}\bra{\psi^*}|^2 = 2^{-n} |\bracket{\psi}{\sigma_r}{\psi^*}|^2$.
\end{proof}


\section{Learning stabilizer states}

We now show that Bell sampling can be used to learn stabilizer states efficiently. By a result of~\cite{dehaene02} (see~\cite{vandennest08} for an alternative proof), up to an overall phase every stabilizer state $\ket{\psi}$ can be written in the form
\[ \ket{\psi} = \frac{1}{\sqrt{|A|}} \sum_{x \in A} i^{\ell(x)} (-1)^{q(x)} \ket{x}, \]
where $A$ is an affine subspace of $\F_2^n$, and $\ell,q: \{0,1\}^n \rightarrow \{0,1\}$ are linear and quadratic (respectively) polynomials over $\F_2$. As $\ell$ is linear, $\ell(x) = s \cdot x$ for some $s \in \{0,1\}^n$, so we have $i^{\ell(x)} = \prod_{k \in S} i^{x_k}$ for some $S \subseteq [n]$. Hence
\[ \ket{\psi^*} = \sigma_{10}^{\otimes S} \ket{\psi}. \]
If we perform Bell sampling on $\ket{\psi}^{\otimes 2}$, by Lemma \ref{lem:bsamp} we receive outcome $r$ with probability
\be \label{eq:stabsamp} \frac{|\bracket{\psi}{\sigma_r}{\psi^*}|^2}{2^n} = \frac{| \bra{\psi} \sigma_r \sigma_{10}^{\otimes S}\ket{\psi} |^2}{2^n}. \ee
Any stabilizer state $\ket{\psi}$ is uniquely specified by a commuting subgroup $G$ of Pauli matrices $M$ (with potentially additional overall phases $\pm1$) such that $|G|=2^n$, $M\ket{\psi} = \ket{\psi}$ for all $M \in G$, and $\bracket{\psi}{M}{\psi} = 0$ for all Pauli matrices $M \notin G$. Let $T$ denote the set of strings $t \in \{0,1\}^{2n}$ such that $\sigma_t \in G$, up to a phase. Then $T$ is an $n$-dimensional linear subspace of $\F_2^{2n}$. Determining $T$ suffices to uniquely determine $\ket{\psi}$: although $T$ does not contain information about phases, once we have found a basis for $T$, we can measure $\ket{\psi}$ in the eigenbasis of each corresponding Pauli matrix $M$ to decide whether $M\ket{\psi} = \ket{\psi}$ or $M\ket{\psi} = -\ket{\psi}$.

By eqn.\ (\ref{eq:stabsamp}), Bell sampling gives an outcome $r$ which is uniformly distributed on the set $\{t \oplus s:t \in T\}$ for some $s \in \{0,1\}^{2n}$. Thus, for any two such outcomes $r_1$, $r_2$, the sum $r_1 \oplus r_2$ is uniformly distributed in $T$. In order to find a basis for $T$, we can therefore produce $k+1$ Bell samples $r_0,r_1,\dots,r_k$, for some $k$, and consider the uniformly random elements of $T$ given by $r_1 \oplus r_0,r_2 \oplus r_0,\dots,r_k \oplus r_0$. If the dimension of the subspace of $\F_2^{2n}$ spanned by these vectors is $n$, any basis of this subspace is a basis for $T$.

We give an explicit description of this algorithm as Algorithm \ref{alg:stabilizer} (boxed). The algorithm uses $5n+2$ copies of $\ket{\psi}$. The time complexity of the algorithm is dominated by the basis-determination step, which can be achieved using Gaussian elimination in time $O(n^3)$; technically, this can be improved to $O(n^\omega)$, where $\omega < 2.373$ is the matrix multiplication exponent. Note that any algorithm for learning a stabilizer state requires time $\Omega(n^2)$ just to write the output.

\boxdfn{Learning stabilizer states}{
\begin{enumerate} \vspace{-11pt}
\item Set $S = \emptyset$.
\item Create two copies of $\ket{\psi}$ and perform Bell sampling, obtaining outcome $r_0$.
\item Repeat the following $2n$ times:
\begin{enumerate}
\item Create two copies of $\ket{\psi}$ and perform Bell sampling, obtaining outcome $r$.
\item Add $r \oplus r_0$ to $S$.
\end{enumerate}
\item Determine a basis for $S$; call this basis $B$.
\item For each element of $B$, measure a copy of $\ket{\psi}$ in the eigenbasis of the corresponding Pauli matrix $M$ to determine whether $M\ket{\psi} = \ket{\psi}$ or $M\ket{\psi} = -\ket{\psi}$.
\end{enumerate}
\label{alg:stabilizer}
}

The algorithm fails (i.e.\ does not identify $\ket{\psi}$) if each of the $2n$ samples $r \oplus r_0$ lies in a subspace of $T$ of dimension at most $n-1$. The probability that the samples are all contained in any one such subspace is $2^{-2n}$; by a union bound over all subspaces of dimension $n-1$, the algorithm fails with probability at most $2^{-n}$.

Algorithm \ref{alg:stabilizer} can be seen as a generalisation of a result of R\"otteler~\cite{roetteler09} which gives an $O(n)$-query algorithm for learning functions $f:\{0,1\}^n \rightarrow \{0,1\}$ which are polynomials of degree 2 over $\F_2$. The algorithm of~\cite{roetteler09} works by producing states of the form
\[ \ket{\psi} = \frac{1}{\sqrt{2^n}} \sum_{x \in \{0,1\}^n} (-1)^{f(x)} \ket{x}, \]
and then proceeds in a similar way to Algorithm \ref{alg:stabilizer} (although it is presented differently).



{\bf Acknowledgements.} This work was largely carried out while the author was at the University of Cambridge, and was supported by the UK EPSRC (EP/G049416/2, EP/L021005/1). Thanks to Joe Fitzsimons for pointing out ref.~\cite{zhao16}.

\bibliographystyle{plain}
\bibliography{../../thesis}

\end{document}